\documentclass{article}

\usepackage{graphicx}

\usepackage{hanmath,amsmath,amssymb,amscd,upgreek}

\newcommand\qedsymbol{\hbox{\rlap{$\sqcap$}$\sqcup$}}
\newenvironment{proof}{\textbf{Proof.}}{\\ \mbox{ }\hfill\qedsymbol\\}

\renewcommand{\im}{\mathop{\mathrm{Im}}}
\renewcommand{\re}{\mathop{\mathrm{Re}}}

\newcommand{\D}{\mathrm{D}}

\renewcommand{\Sigma}{\mathfrak{S}}

\newcommand{\I}{\mathrm{I}}

\newcommand{\upi}{\uppi}

\author{M. Seredy\'{n}ska\footnote{Corresponding Author}\\
Institute of Fundamental Technological Research\\
Polish Academy of Sciences\\
 ul. Pawi\'{n}skiego 5b\\
02-106 Warszawa, PL\\
email: {\tt msered@ippt.gov.pl}\\ \mbox{ }\\
Andrzej Hanyga\\
freelance scientist\\
ul. Bitwy Warszawskiej 14 m. 52,\\
02-366 Warszawa, PL\\
email: {\tt ajhbergen@yahoo.com}} 

\title{Attenuation, dispersion and finite propagation speed in
viscoelastic media}

\date{}

\begin{document}

\maketitle

\begin{abstract}
It is shown that the dispersion and attenuation functions in a linear 
viscoelastic medium 
with a positive relaxation spectrum have a sublinear growth rate at very high 
frequencies. A local dispersion relation in parametric form is found. The exact limit 
between attenuation growth rates compatible and incompatible with finite propagation speed
is found. 
 
Incompatibility
of superlinear frequency dependence of attenuation with finite speed of 
propagation and with the assumption of positive relaxation spectrum is 
demonstrated. 
\end{abstract}

\vspace{0.5cm}

\noindent\textbf{List of symbols.}\\
\begin{tabular}{lll}
$\tilde{f}(p)$ & $\int_0^\infty f(t) \, \e^{-p t} \, \dd t$ & Laplace transform;\\
$\theta(t)$ &  & unit step function;\\
$t_+^\alpha$ & $\theta(t) \vert t \vert^\alpha$ & homogeneous distribution.\\
$\D^\alpha_{\mathrm{C}}$ & $\I^{n-\alpha}\, \D^n$ & Caputo fractional derivative
\end{tabular}

\section{Introduction.}

A very elegant theory 
linking attenuation and dispersion is presented for viscoelastic media 
with positive relaxation spectrum. 
All the attenuation and dispersion functions compatible with the theory 
are represented by simple integral representations. The attenuation function and
the dispersion function are both expressed as transforms of a positive measure
(the dispersion-attenuation measure). The measure is arbitrary except for a 
very mild growth condition. These expressions can be considered as a 
dispersion relation in parametric form. 

In contrast the acoustic Kramers-Kronig dispersion \cite{WeaverPao81} relations are non-local. 
They express the dispersion function in terms of the attenuation function or 
conversely. This presupposes that one of these functions (usually the 
dispersion function) is very accurately known and consistent with the basic 
assumptions of the theory. On the other hand, substituting any positive 
measure in the parametric dispersion relation yields an admissible
(compatible with the theory) dispersion and attenuation. The origin of 
the Kramers-Kronig dispersion relations is unclear. In electromagnetic 
theory they follow from the causality of the time-domain kernel representing 
the dielectric constant in the dispersive case.
In acoustics they follow from an ad hoc assumption about the analytic properties
of the wave number. It is namely assumed that the wave number is the 
Fourier transform of a causal function or distribution. The physical meaning 
of the causal function or distribution is unclear hence the justification of 
the acoustic Kramers-Kronig dispersion relation is missing. Various 
inequalities imposed on the complex continuation of the wave number cannot 
be expressed in terms of a constitutive assumption. Hence
the acoustic Kramers-Kronig dispersion relation are an ad hoc addition to the 
constitutive equation, often incompatible with it.

It will be also shown that in media with non-negative 
relaxation spectrum the frequency dependence of the attenuation function 
in the high frequency range is sublinear. In the case of power law attenuation 
the attenuation and dispersion are proportional to a power of frequency 
$\vert \omega\vert^\alpha$. Numerous experiments in acoustics indicate
that the power-attenuation law accurately represents the frequency dependence
of dispersion and attenuation over several decades of frequency. 
The theory based on positive relaxation spectrum implies that 
$0 \leq \alpha < 1$.  
This seems to contradict the experimental ultrasound investigations of 
numerous materials which point to higher
values of the exponent in the power law attenuation. For example in 
ultrasound investigations of soft 
tissues the exponent varies between 1 and 1.5, while in some viscoelastic fluids
such as castor oil it lies between 1.5 and 2.
We shall refer to this case as superlinear frequency dependence.
Typical values of the power-law exponent in medical applications using 
ultrasound transducers 
are $\alpha = 1.3$ in bovine liver for 1--100~MHz, $\alpha = 1-2$ in
human myocardium and other bio-tissues \cite{SzaboWu00}.
$\alpha = 1.66$ in castor oil at ca 250~MHz. Values in the range 1--2
are observed at lower frequencies in aromatic polyurethanes 
\cite{GuessCampbell95}. Nearly linear frequency dependence of attenuation 
is well documented in seismology \cite{Futterman}. Approximately 
linear frequency dependence of attenuation has been observed in 
geological materials in the range 140~Hz to 2.2~MHz. 

Several papers have been
devoted to a theoretical underpinning of the superlinear dispersion-attenuation 
models 
\cite{Szabo1,Szabo2,ChenHolm03,WatersHughersBrandenburgerMiller,CobboldSushilovWeathermon04,KellyMcGoughMeerschaert08}.
In order to resolve some problems Chen and Holm \cite{ChenHolm04} suggested to add a fractional Laplacian
of order $0 < y < 2$ of the velocity field to the usual Laplacian of the displacement 
field in the equations of motion.  Their paper still leads to an unbounded sound speed for $y > 1$ 
and adds a new problem: the equation of motion does not have the form of a viscoelastic 
equation of motion. 

Sublinear power law attenuation 
($\alpha < 1$) has also been reported in experimental investigations  
\cite{RibodettiHanygaGJI}. Non-power law attenuation laws are usually 
derived from the constitutive laws. Investigations of creep and relaxation 
in viscoelastic materials 
always support the assumption of positive relaxation spectrum 
(e.g. \cite{Andrade1,Andrade2} for creep in metals, \cite{MinsterAnderson}
for the upper mantle with $\alpha = 1/3$) and therefore models derived 
from constitutive relations exhibit sublinear attenuation and dispersion 
at high frequencies. In particular this applies to the Cole-Cole, 
Havriliak-Negami and 
Cole-Davidson and Kohlrausch-Williams-Watts relaxation laws commonly applied 
in phenomenological rock mechanics, polymer
rheology \cite{BagleyTorvik4}, bio-tissue mechanics (e.g. for bone 
collagen \cite{SasakiAl93}) as well as for ionic glasses \cite{CariniAl84}.

Another abnormal feature of wave propagation in media with superlinear 
power laws (and more generally in media with superlinear asymptotic 
frequency dependence) is appearance of precursors. The precursors extend 
to infinity and thus the 
speed of propagation of disturbances is infinite. Finite speed of wave 
propagation requires that in the high-frequency range the exponent of the 
power law does not exceed 1. 

It is a very challenging problem how to explain the incompatibility 
between the theory and experiment in the superlinear case. It seems likely 
that the attenuation 
observed at ultrasound frequencies significantly differs from the
asymptotic behavior of attenuation at the frequency tending to infinity. 

One might surmise that the frequency range of ultrasound measurements is still
far below the asymptotic high frequency range in which different mechanisms
are at play. This suggests studying models of attenuation with a slowly 
varying power-law exponent.

\section{Constitutive assumptions and basic definitions.}

In viscoelasticity the relaxation modulus $G$, 
defined by the constitutive stress-strain relation
\begin{equation} \label{eq:0}
\upsigma(t) = \int_0^t G(t-s) \, \dot{e}(s)\, \dd s
\end{equation}
is assumed to have positive relaxation spectral density $h$. The latter 
statement means that for every $t > 0$ 
\begin{equation} \label{eq:1}
G(t) = \int_0^\infty \e^{-t r}\, h(r)\, \dd \ln(r), \qquad t > 0
\end{equation}
where $r = 1/\tau$ is the inverse of the relaxation time and $h(r) \geq 0$. 
Eq.~\eqref{eq:1} represents 
a superposition of a continuum of Debye elements. For mathematical convenience 
eq.~\eqref{eq:1} will be replaced by a more general equation
\begin{equation} \label{eq:2}
G(t) = \int_{[0,\infty\;[} \e^{-t r}\,\mu(\dd r), \qquad t > 0
\end{equation} 
where $\mu$ is a positive measure:
$\mu([a,b]) \geq 0$ for every interval $[a,b]$ of the positive real axis. 
As indicated in the subscript of the integral sign the range of integration 
is the set of reals satisfying the inequality $0 \leq r < \infty$. In general 
the measure $\mu(\{0\})$ of the one-point set $\{0\}$ is finite and equal
to the equilibrium modulus $G_\infty := \lim_{t\rightarrow\infty} G(t)$. 
An additional assumption
\begin{equation}
\int_{[0,\infty[\;} \frac{\mu(\dd r)}{1 + r} < \infty
\end{equation} 
ensures that $G$ is integrable over $[0,1]$, \cite{HanDuality}. The 
relaxation modulus assumes 
a finite value $G(0) = M$ at 0 if the measure $\mu$ has a finite mass $M$.

The right-hand side of eq.~\eqref{eq:2} 
can be replaced by a Stieltjes integral with respect to the function $g(r) = 
\mu([0,r])$:
\begin{equation} \label{eq:3}
G(t) = \int_{[0,\infty\;[} \e^{-t r}\,\dd g(r)
\end{equation} 
The function $g$ is non-decreasing and right-continuous: 
$g(r) = \mu([0,r]) = \lim_{\varepsilon \rightarrow 0+} \mu([0,r+\varepsilon])$,
and $g(r) = 0$ for $r < 0$.

In contrast to eq.~\eqref{eq:1} the integral representations \eqref{eq:2} 
and \eqref{eq:3} include as special cases finite spectra of relaxation times 
corresponding to superpositions of a finite number of Debye 
elements 
\begin{equation} \label{eq:4}
G(t) = \sum_{n=1}^N c_n \, \e^{-r_n\, t}, \qquad c_n > 0, r_n \geq 0 
\quad\text{for $n = 1,\ldots,N$}
\end{equation}
(Prony sums), infinite discrete spectra (corresponding to Dirichlet series) as well 
as discrete spectra embedded in continuous spectra.
However the main advantage of \eqref{eq:2} over \eqref{eq:1} is the 
availability of very powerful mathematical theory which ensures logical 
equivalence of certain statements about material response functions. 

In particular a function satisfying \eqref{eq:2} is completely monotone. A function $G$ is
said to be \emph{completely monotone} if it continuously differentiable to any order and
\begin{equation} \label{eq:CM}
(-1)^n \, \D^n \, G(t) \geq 0 \qquad \text{for all non-negative integers 
$n$ and $t > 0$}
\end{equation}
Bernstein's theorem \cite{WidderLT,Jacob01I} 
asserts that eq.~\eqref{eq:2} is {\em equivalent} to \eqref{eq:CM}.
There is no such simple characterization of functions which have the form
\eqref{eq:1}. 

It was established in \cite{Bland:VE,Molinari,HanDuality} that viscoelastic 
relaxation moduli are completely monotone.

For us the main benefit from using \eqref{eq:2} instead of \eqref{eq:1} 
is the equivalence 
of \eqref{eq:2} with a property of the dispersion and attenuation that will be
explained below. That is, certain statements about attenuation and dispersion 
functions follow from \eqref{eq:2} or, conversely, imply that \eqref{eq:2}
does not hold.

We now recall the definition of a Bernstein function \cite{BergForst}. A 
differentiable 
function $f$ on $]0,\infty[$ is a Bernstein function if $f \geq 0$ and 
its derivative 
$f^\prime$ is completely monotone. A Bernstein function is non-negative, 
continuous 
on $]0,\infty[$ and non-decreasing, hence it has a finite value at 0. A 
function $f$ on $[0,\infty[$ which has the form $f(x) = x^2\, \tilde{g}(x)$ 
for some Bernstein function $g$ is called a \emph{complete Bernstein function} (CBF) 
\cite{Jacob01I}. It can be
proved that every complete Bernstein function is a Bernstein function
\cite{Jacob01I}. 

The following facts about complete Bernstein functions will be needed here. 
\begin{theorem} \label{thm:J}
A real function $f$ on $]0,\infty[$ is a complete Bernstein function if (i) $f$ 
has an analytic continuation $f(z)$
to the complex plane cut along the negative real axis; (ii) $f(0) \geq 0$, (iii) 
$f\left(\overline{z}\right) = \overline{f(z)}$, (iv) $\im f(z) \geq 0$ in the
upper half plane $\im z \geq 0$.
\end{theorem}
Theorem~\ref{thm:J} has the following corollaries: 
(1) If $f$ is a CBF then $f^\alpha$ is a CBF if $0 \leq \alpha \leq 1$
\cite{HanSer09}. \\
(2) If $f \not\equiv 0$ then the function $x/f(x)$ is a CBF. 

Every complete Bernstein function $f$ has the integral representation:
\begin{equation}
f(x) = a + b\, x + x \int_{]0,\infty[\;} \frac{\nu(\dd r)}{x + r}
\end{equation}
where $a, b \geq 0$ and $\nu$ is a positive measure satisfying the inequality
\begin{equation} \label{eq:5}
\int_{[0,\infty[\;} \frac{\nu(\dd r)}{1 + r} < \infty
\end{equation}
\cite{Jacob01I}.

If $G$ is a completely monotone function integrable over $[0,1]$ and 
satisfying eq.~\eqref{eq:2} then 
$$\tilde{G}(p) = \int_{[0,\infty[\;} \frac{\mu(\dd r)}{p + r} =
\frac{\mu(\{0\})}{p} + \int_{]0,\infty[\;} \frac{\mu(\dd r)}{p + r}$$
It follows that
\begin{equation} \label{eq:8} 
\mathcal{Q}(p) := p \, \tilde{G}(p) = \mu(\{0\}) + \int_{]0,\infty[\;} 
\frac{\mu(\dd r)}{p + r}
\end{equation}
is a complete Bernstein function \cite{HanSer09}.

\section{Dispersion and attenuation.}

The Green's function $\mathcal{G}(t,x)$ is defined as the solution of 
the problem
\begin{equation}
\rho u_{,tt} = G(t)\ast u_{,txx} + \delta(x)\, \delta(t)\\
\end{equation}
with zero initial data.
In a three-dimensional space 
\begin{multline}  \label{eq:3Dgreen}
\mathcal{G}(t,x) = \frac{-1}{(2 \upi)^3\, r} \int_{-\ii \infty + \varepsilon}^{\ii \infty + \varepsilon} \e^{p t} \, \frac{1}{\mathcal{Q}(p)}\,\dd p \int_{-\infty}^\infty \, 
\frac{\e^{\ii k r}}{k^2 + B(p)^2}\, k \, \dd k = \\ 
\frac{1}{8 \upi^2 \,\ii\, r} \int_{-\ii \infty + \varepsilon}^
{\ii \infty + \varepsilon} \,\frac{1}{\mathcal{Q}(p)}
\e^{p\, t - B(p) \, r}\, \dd p 
\end{multline}
where
\begin{equation} \label{eq:6} 
B(p) := p\,\frac{\rho^{1/2}}{\mathcal{Q}(p)^{1/2}} 
\end{equation}
If $G$ is completely monotone then $\mathcal{Q}$ is a complete Bernstein 
function.
If a function $f$ is a complete Bernstein function then the functions 
$f^{1/2}$ and $p/f(p)$ are complete Bernstein functions. Hence $B$ is a 
complete Bernstein function and
$B(p) = a + c\, p + b(p)$ where $a, c \geq 0$ and the dispersion-attenuation function 
$b$ has the integral representation
\begin{equation} \label{eq:attenuationspectrum}
b(p) := p\, \int_{]0,\infty[\;} \frac{\nu(\dd r)}{p + r}
\end{equation}
and $\nu$ is a positive measure satisfying \eqref{eq:5}. 
A function $b$ having an integral representation \eqref{eq:attenuationspectrum}
will be called an \emph{admissible dispersion-attenuation function}.
The measure $\nu$ is the spectral measure of the dispersion-attenuation function $b$. 
A dispersion-attenuation function $b$ of a viscoelastic material with a positive 
relaxation spectrum is admissible. Conversely, for any admissible 
dispersion-attenuation function there is a completely monotone relaxation
modulus $G$ satisfying eq.~\eqref{eq:6} and eq.~\eqref{eq:8}. 

Using this equivalence it is possible to decouple the study of admissible 
dispersion-attenuation functions
from considering specific constitutive equations. Furthermore, in order to 
avoid incompatibility with viscoelastic theory experimental 
studies should target the spectral measure $\nu$ of $b$ rather than the 
dispersion-attenuation function 
$b(p)$. Any positive measure $\nu$ satisfying \eqref{eq:5} is compatible with
the theory and the implications of the theory for the dispersion-attenuation function
$b$ are encapsulated in \eqref{eq:attenuationspectrum}. Using the spectral measure
$\nu$ instead of $b$ to match the experimental data 
is analogous to expressing relaxation modulus in terms of the 
relaxation spectrum, e.g. by applying Prony sums to represent experimental 
data. A discretization of the dispersion-attenuation measure 
is presented in Sec.~\ref{ssec:discrete}.

Let $G_\infty = \lim_{t \rightarrow \infty} G(t)$. Since $\lim_{p\rightarrow 0} 
\left[ p \,\tilde{G}(p) \right] = G_\infty$, 
the inequality $G_\infty > 0$ and eq.~\eqref{eq:6} imply that 
$B(0) = 0$ and therefore $a = 0$. The limit 
$\lim_{p\rightarrow\infty} B(p)/p = c$ implies that $c = 1/c_0$, where $c_0$ 
denotes the wave front speed.  The last conclusion follows from the
fact that for $p > 1$ 
$$\frac{1}{p + r} \leq \frac{1}{1 + r}$$
and the right-hand side is integrable with respect to the measure $\nu$ in view of  
eq.~\eqref{eq:5}. For $p \rightarrow \infty$ the integrand of 
$$\int_{]0,\infty[\;} \frac{\nu(\dd r)}{p + r}$$
tends to zero, hence, by the Lebesgue Dominated Convergence Theorem, 
the integral tends to zero as well. It follows additionally that 
$b(p) = \mathrm{o}[p]$ for $p \rightarrow \infty$. 
Hence \eqref{eq:3Dgreen} can be recast in a more explicit form 
\begin{equation} \label{eq:Green3D}
\mathcal{G}(t,x) = \frac{1}{8 \upi^2 \,\ii\, r} \int_{-\ii \infty + \varepsilon}^
{\ii \infty + \varepsilon} \,\frac{1}{\mathcal{Q}(p)}
\e^{p\,( t - \vert x \vert/c_0) - b(p) \, r}\, \dd p 
\end{equation}
and the dispersion-attenuation function $b(p)$ has sublinear growth.

We also note that the attenuation function is non-negative: 
\begin{equation}
\mathcal{A}(p) := \re b(p) = \int_{]0,\infty[\;} 
\frac{\vert p \vert^2 + r \, \re p}
{\vert p + r\vert^2} \nu(\dd r) \geq 0 
\end{equation}
for $\re p \geq 0$. 

The derivative 
\begin{equation}
b^\prime(p) = b(p)/p - \int_{]0,\infty[\;} 
\frac{\nu(\dd r)}{(p + r)^2} \leq b(p)/p
\end{equation}
Moreover 
$$b^\prime(p) = \int_{]0,\infty[\;} \frac{r\, \nu(\dd r)}{(p + r)^2} 
\geq 0$$
Hence $b^\prime(p) = \mathrm{o}[1]$ for $p \rightarrow \infty$.

The attenuation function $\mathcal{A}$ and the dispersion function 
$\mathcal{D}(p) :=
-\im b(p)$ satisfy linear dispersion equations in parametric form:
\begin{gather}
\mathcal{A}(p) = \int_{]0,\infty[\;} 
\frac{\vert p \vert^2 + r \, \re p}
{\vert p + r \vert^2} \nu(\dd r)\\
\mathcal{D}(p) = -\im p \int_{]0,\infty[\;} 
\frac{r}{\vert p + r\vert^2} \nu(\dd r)
\end{gather}

The measure $\nu$ represents the dispersion-attenuation spectrum. An elementary 
dispersion-attenuation is represented by the function $p/(p + r)$ for a fixed value of $r$.

\section{Implications of finite speed of wave propagation.}

Since $1/\mathcal{Q}(p) = p \, \tilde{J}(p) = J(0) + \widetilde{J^\prime}(p)$,
where $J$ is the creep compliance, we can express the Green's function \eqref{eq:Green3D}
in the form of a convolution
$$\mathcal{G}(t,x) = J^\prime(t)\ast H(t-\vert x \vert/c_0,\vert x\vert)
+ J_0 \, H(t-\vert x \vert/c_0,\vert x\vert)$$
where $H(\tau,r)$ is the inverse Laplace transform of the function
$\e^{-b(p)\,r}/(4 \upi r)$. In terms of the inverse Fourier 
transformation
$$H(\tau,r) = \frac{1}{8 \upi^2 \, r} \int_{-\infty}^\infty
\e^{-\ii \omega \tau - b(-\ii \omega)\, r} \, \dd \omega$$
Note that $b(-\ii \omega) = \overline{b(\ii \omega)}$.
The integrand is square integrable if $r > 0$ and 
\begin{equation} \label{eq:condPW}
\int_0^\infty \e^{-2 \re b(-\ii \omega)\, r} \dd \omega < \infty
\end{equation}

If eq.~\eqref{eq:condPW} holds then 
the Paley-Wiener theorem (Theorem~XII in \cite{PaleyWiener}) can be applied.
By this theorem  
$H(\tau,r) = 0$ for $\tau < 0$
if and only if 
\begin{equation} \label{eq:PW}
\int_{-\infty}^\infty \frac{\re b(-\ii \omega)}{1 + \omega^2} \dd \omega < \infty
\end{equation}
$J^\prime$ is a causal function. Hence, 
if $H(\tau,r) = 0$ for $\tau < 0$ then 
$$\mathcal{G}(t,x) = \int_0^{t-\vert x \vert/c_0} 
J^\prime(s) \, H(t - \vert x \vert/c_0 - s, \vert x \vert) \, \dd s $$
vanishes for $t < \vert x \vert/c_0$.

\begin{figure}
\includegraphics[width=0.75\textwidth]{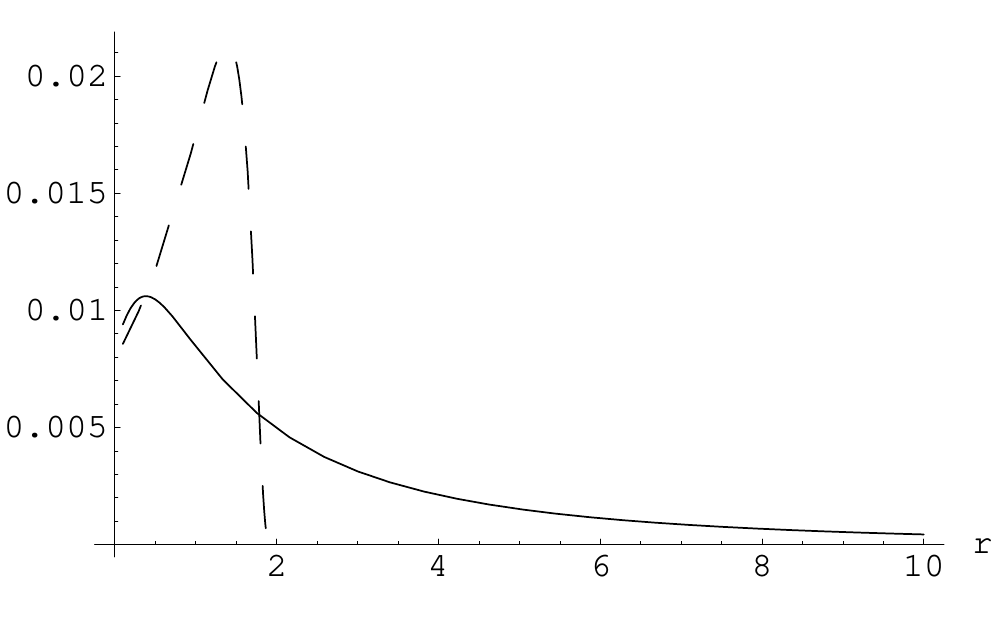}
\caption{Two snapshots of Green's functions for power-law attenuation with 
$\alpha=3/2$ (solid line) and $\alpha=1/2$ (dashed line).} 
\label{fig:1}
\end{figure}

In particular, if $b(p) \sim_\infty a\, p^\alpha$, $a > 0$, $\re p \geq 0$,
eqs~\eqref{eq:condPW} and \eqref{eq:PW} are ensured by the inequality 
$0< \alpha < 1$. A precursor appears for $\alpha > 1$, as can be seen from 
Fig.~\ref{fig:1}. For $\alpha = 3/2$\footnote{For $\alpha=1/3, 2/3$ the
$\alpha$-stable probability can be expressed in terms of Airy functions, see 
\cite{HanQAM}.} there is no wavefront and the 
peak is preceded by a precursor extending to infinity. The limit case, as we already know,
is the asymptotic behavior $b(p) \sim_\infty a\,p/\ln(p)$ and it entails
unbounded propagation speed.

For a general viscoelastic medium with a positive relaxation spectrum we note
that
\begin{equation}
\re b(-\ii \omega) = \omega \im \int_{]0,\infty[\;} \frac{\nu(\dd \xi)}{\xi - \ii \omega} 
= \omega^2 \int_{]0,\infty[\;} \frac{\nu(\dd \xi)}{\xi^2 + \omega^2}
\end{equation} 
If the total mass $M$ of  $\nu$ is finite then 
\begin{equation}
\lim_{\omega\rightarrow\infty}\;\re b(-\ii \omega) = M
\end{equation}
by the Lebesgue Dominated Convergence Theorem.
In the general case the inequality
$$\frac{\omega}{\xi^2 + \omega^2} \leq \frac{2 \omega}{\xi^2 + 2 \xi \omega + \omega^2} 
\leq \frac{2 \omega}{(1 + \xi) \, (\xi + \omega)}$$
valid for $\omega \geq 1$ and the Lebesgue Dominated Convergence Theorem imply that
$\re b(-\ii \omega) = \oo[\omega]$ at $\omega \rightarrow \infty$. 

We have thus proved that the asymptotic growth of 
$\re b(-\ii \omega)$ in the high-frequency range is sublinear.
This is however insufficient to ensure convergence of the left-hand side of \eqref{eq:PW}
and vanishing of the wave field ahead of the wavefront $\vert x \vert = c_0\, t$. 
For example for $b(-\ii \omega) \sim_\infty  \omega/\ln^\alpha(\omega)$ 
the integral in \eqref{eq:PW} does not always converge. Since we are interested in
the high-frequency behavior, we can replace $1 + \omega^2$ by $\omega^2$ in the
denominator of \eqref{eq:PW}. 
For $\alpha = 1$ 
$$\int_\e^\omega \frac{\dd y}{y\, \ln(y)} = \ln(\ln(\omega))$$
is unbounded for $\omega \rightarrow \infty$, hence \eqref{eq:PW} is not satisfied. 
For $\alpha \neq 1$
$$\int_\e^\omega \frac{\dd y}{y\, \ln^\alpha(y)} = 
\frac{\ln^{1-\alpha}(\omega)-1}{1 - \alpha}$$
is unbounded for $\alpha < 1$ and bounded for $\alpha > 1$. 
We thus see that in contrast to the power law attenuation, linear growth is not a limit
case for finite propagation speed. Some sublinear cases will also exhibit precursors 
ahead of the wave front. 

In particular the function $b(p) = p/\ln^\alpha(1 + p)$ with $0 \leq \alpha \leq 1$ is a 
CBF with the 
asymptotic properties discussed in the previous paragraph. Indeed,
for $\im p \geq 0$ the argument $\psi = \arg(1 + p)$ satisfies the inequality
$0 \leq \psi \leq \upi$. Hence $\ln(1 + p)$ maps the upper half plane into itself and is
non-negative for $p \geq 0$. Hence, by Theorem~\ref{thm:J}, $\ln(1 + p)$ is a CBF.
The same is true for $\ln^\alpha(1 + p)$ if $0 < \alpha \leq 1$.  In view of a property of CBF functions mentioned after Theorem~\ref{thm:J} this implies
that $b(p)$ is a CBF if $0 < \alpha \leq 1$. Moreover $b(0) = 0$ and
$\lim_{p\rightarrow \infty} b(p)/p = 0$, hence $b(p)$ is an admissible 
dispersion-attenuation function. 
We thus have produced an example of an admissible dispersion-attenuation function $b(p)$ such that 
$b(-\ii \omega)$ has sublinear growth in the
high frequency range but \eqref{eq:PW} is not satisfied.

\section{Examples.}
\subsection{Discrete dispersion-attenuation spectra.}
\label{ssec:discrete}

If the measure $\nu = \sum_{n=1}^N c_n\, \varepsilon_{r_n}$, $r_n, c_n > 0$,
where $\varepsilon_a = \delta(r - a)$ denotes the Dirac measure concentrated at the point $a$,
then
\begin{equation} \label{eq:7}
b(p) = p \sum_{n=1}^N c_n/(p + r_n)
\end{equation} 

Eq.~\eqref{eq:7} can be used to construct numerical approximations of experimental data.
Using the relation $\mathcal{Q}(p)\, B(p)^2 = \rho\, p^2$ and the assumption 
$\mu = \sum_{m=1}^M d_m\, \varepsilon_{s_m}$ with $d_m > 0$, $s_m \geq 0$ for $1 \leq m \leq M$,
$s_m \neq r_n$ for all $n, m$, it is possible to express the relaxation data $d_m, s_m$
in terms of the dispersion-attenuation data $c_n, r_n$ and conversely.

\subsection{Power-law attenuation.}

Power-law attenuation is commonly used to match experimental dispersion and attenuation 
data for a wide variety of real viscoelastic materials such as polymers, bio-tissues and some 
viscoelastic fluids. 

Consider the viscoelastic medium defined by the following equation of motion
\begin{equation} \label{eq:Rok}
\rho\, \left(\D_\mathrm{C}^2 + 2 a\, \D_\mathrm{C}^{1+\alpha} + 
a^2\, \D_\mathrm{C}^{2 \alpha} \right) u = A\, \nabla^2 \, u + \delta(t)\, \delta(x)
\end{equation}
with the initial data $u(0,x) = u_{,t}(0,x) = 0$ in $d$ dimensions, $d =1,3$.
$\D_\mathrm{C}^\alpha$ denotes the Caputo fractional derivative of order 
$\alpha$ 
\begin{equation}
\D_\mathrm{C}^\alpha f(x) := \int_0^t \frac{(t-s)^{-\alpha}}{\Gamma(1-\alpha)} f^\prime(s) \, \dd s
\end{equation}
\cite{PodlubnyBook}.

It is assumed that $A > 0$, $a \geq 0$, $0 < \alpha < 1$.
The Laplace transformation $\mathcal{L}_t$ with respect to the time variable and the 
Fourier transformation $\mathcal{F}_{x}$
with respect to the spatial variable bring eq.~\eqref{eq:Rok} to the following form
\begin{equation} \label{eq:Rok1}
\rho\,g(p)^2 \, \hat{u}(p,\kk) 
= -A\,\kk^2\,\hat{u}(p,\kk) + 1 
\end{equation}
where $\hat{u} := \mathcal{F}_{x}\left(\mathcal{L}_t(u)\right)$ and 
$g(p) := p + a\, p^\alpha$. Eq.~\eqref{eq:Rok1} implies that in the power law attenuation model
$\mathcal{Q}(p) = p\,A/g(p)^2$ and $B(p) = g(p)/c_0$.

We shall begin with solving eq.~\eqref{eq:Rok} in one dimension. Applying the inverse
Fourier transformation to eq.~\eqref{eq:Rok1}:
$$\tilde{u}^{(1)}(p,x) = \frac{1}{2 \upi} \int_{-\infty}^\infty 
\frac{\e^{\ii k x + p t}}{\rho \,g(p)^2 + A\, k^2} \dd k$$
The contour can be closed by a large half-circle in upper-half complex $k$-plane if 
$x > 0$, in the lower- half complex $k$-plane if $x < 0$ and 
$$\frac{1}{A\, k^2 + \rho\,g(p)^2} = \frac{1}{2 \ii g(p)\, A} 
\left[\frac{1}{k - \ii g(p)/c_0} - \frac{1}{k + \ii g(p)/c_0}  \right]$$
where $c_0 := \sqrt{A/\rho}$. We now restrict ourselves to imaginary values of 
$p = -\ii w$. Since $\im [\ii\, g(-\ii w)] = \re g(-\ii w) = a\,\vert w \vert^\alpha \,
\cos((1-\alpha)\,\upi/2) \geq 0$, the residuum at $\pm \ii g(p)/c_0$ contributes if 
$\pm x > 0$. Hence
$$\tilde{u}(p,x) = \frac{1}{2 A} 
\frac{\e^{-g(p) \vert x \vert/c_0 + p t}}{g(p)} $$
and
\begin{equation}
u^{(1)}(t,x) = \frac{1}{4 \upi \ii \, A} \int_{-\ii \infty}^{\ii \infty} \frac{1}{g(p)}
\e^{p \,t - g(p)\vert x \vert/c_0} \dd p
\end{equation}
The solution $u^{(3)}$ of \eqref{eq:Rok} in a three-dimensional space 
is given by the formula
$$u^{(3)}(t,x) = \frac{-1}{2 \upi r} \frac{\partial}{\partial r} 
u^{(1)}(t,r)$$
Hence 
\begin{multline}
u^{(3)}(t,x) = \frac{1}{8 \upi^2 \ii \vert x \vert\, c_0\, A}  
\int_{-\ii \infty}^{\ii \infty} \e^{p \,(t - \vert x \vert/c_0) - 
a \, p^\alpha\, \vert x \vert} \dd p = \\
\frac{1}{4 \upi a^{1/\alpha}\,\vert x \vert^{1+1/\alpha}\, c_0 \, A}\, 
P_\alpha\left((t - \vert x \vert/c_0)/(a \vert x \vert)^\alpha\right)
\end{multline}
where
\begin{equation}
P_\alpha(z) := \frac{1}{2 \upi \ii} \int_{-\ii \infty}^{\ii \infty} 
\e^{y z} \, \e^{-y^\alpha}\, \dd y
\end{equation}
The function $P_\alpha$ is a totally skewed L\'{e}vy stable probability density 
\cite{Sato:Levy,Zolotarev2}. 

In the case of power-law attenuation $\re b(-\ii \omega) =
a\, \omega^\alpha \, \cos(\upi \alpha/2) > 0$ and for $p = -\ii \omega$ 
$$\left\vert \e^{p\, t - g(p)\, r/c_0} \right\vert = 
\e^{-a\,\omega^\alpha\, \cos(\alpha\,\upi/2)}$$
Hence (\cite{WhittakerWatson}, Sec.~4.44) the integrals
$$f(t,r) := \frac{1}{2 \upi \ii} \int_\mathfrak{B} (-1)^n \,p^{-2}\,p^m\, g(p)^{n+1} \, 
\e^{p\,t - g(p) \, r/c_0} \, \dd p$$
are uniformly convergent for $r > 0$ and the derivatives
$\partial^{n+m}\, f/\partial t^n\, \partial r^m$ exist for all positive integers $n, m$
and $r > 0$. 

All the properties of the dispersion-attenuation function $b$ were derived from the fact that 
$B(p) = p\,\rho^{1/2}/\mathcal{Q}(p)^{1/2}$ is a complete
Bernstein function. It follows from the above property of $B$ that $\mathcal{Q}(p)^{1/2}$ is a 
complete Bernstein function  \cite{HanSer09}. 
It does not follow from here that $\mathcal{Q}$ is a also complete Bernstein function.
Therefore $G$ need not be a completely monotone function. A counterexample is provided 
by the power law attenuation model.
If $0 \leq \alpha < 1/2$ then the relaxation modulus $G$ is not completely 
monotone but the attenuation function satisfies
eq.~\eqref{eq:attenuationspectrum}.  The latter condition is satisfied for the 
power law attenuation if and only if $0 \leq \alpha < 1/2$ (Fig.~\ref{fig:0}). 

\begin{figure}
\begin{center}
\includegraphics[width=0.75\linewidth]{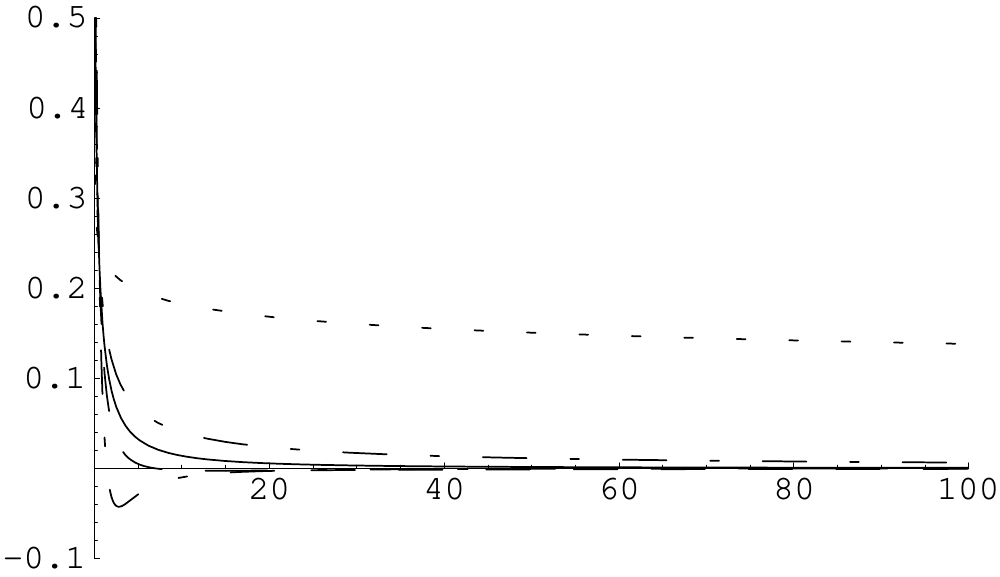}
\end{center}
\caption{The relaxation modulus $G$ for the power-law attenuation. 
The exponent values are $\alpha = 0.3, 0.4, 0.5, 0.6, 0.9$,
from bottom to top.}
\label{fig:0}
\end{figure}
\begin{theorem}
The spectral measure of the dispersion-attenuation function $a\, p^\alpha$, 
$0 \leq \alpha < 1$, is 
\begin{equation} \label{eq:attspepower}
\nu(\dd \xi) = \frac{\sin(\upi \alpha)}{\upi} \xi^{\alpha-1} \, \dd \xi
\end{equation}
\end{theorem}
\begin{proof}
The Laplace transforms
\begin{gather}
x^{\alpha-1} = \int_0^\infty \e^{-x y} \, \left[y^{-\alpha}/\Gamma(1-\alpha)\right]\, \dd y\\
y^{-\alpha} =  \int_0^\infty \e^{-y z} \, \left[z^{\alpha-1}/\Gamma(\alpha)\right]\, \dd z
\end{gather}
hence
$$x^{\alpha - 1} = \frac{\sin(\alpha \upi)}{\upi} \int_0^\infty \frac{z^{\alpha-1}}{x + x} 
\, \dd z$$
which implies eq.~\eqref{eq:attspepower}.
\end{proof}
 
\section{Superlinear power law attenuation.}

A large number of papers have been devoted to the implications of the
Kramers-Kronig dispersion relations for the wave number $K(\omega)$. The 
Kramers-Kronig dispersion relations would follow from the assumption 
that the function $K(\omega) - \omega/c_0$ is the Fourier transform of a causal 
function or 
causal tempered distribution $L(t)$. A priori the function $L(t)$ has no 
physical meaning and the assumption of causality of $L$ is unwarranted.  
Causality of $L$ would however be justified by the assumption
that 
$$B(p) = p/c_0 + \tilde{L}(p) $$
and the equation of motion has the following form 
\begin{equation} \label{eq:motion}
c_0^{\;-2}\, u_{,tt} + (2/c_0)\, L\ast u_{,t} + L\ast L\ast u =
u_{,xx}
\end{equation}
In this case $K(\omega) = \ii B(-\ii \omega) = 
\omega/c_0 + \ii \, \tilde{L}(-\ii \omega) = 
\ii \omega/c_0 + \hat{L}(\omega)$.  Eq.~\eqref{eq:motion} is however
incompatible with the viscoelastic constitutive equation. In a viscoelastic 
equation of motion integral operators should act on the Laplacian of $u$.

In \cite{SzaboWu00} the authors try to guess the viscoelastic constitutive
equation 
compatible with \eqref{eq:motion}. Their approach involves an approximation
of a spatial derivative by a temporal derivative. It is however possible 
to avoid
an approximation by shifting the dispersive terms on the left-hand side 
of \eqref{eq:motion} to the right-hand side. The Laplace transform of the 
left-hand side of eq.~\eqref{eq:motion} is
$$ p^2\, \left[ 1 + c_0\, \tilde{L}(p)/p \right]^2/c_0^{\;2}$$
assuming that $u(0,x) = u_{,t}(0,x) = 0$. 
Hence \eqref{eq:motion} has the form 
$$c_0^{\;-2}\,u_{,tt} = \left[G(t)\ast u_{,tx}\right]_{,x}$$
where $G$ is the inverse Laplace transform of
\begin{equation} \label{eq:xxx}
p^{-1}\,\left[ 1 + c_0\, \tilde{L}(p)/p \right]^{-2} 
\end{equation}
Expression~\eqref{eq:xxx} is the Laplace transform of a completely monotone
function if and only if $\left[1 + c_0\, \tilde{L}(p)/p\right]^2/p$ is the 
Laplace transform of a Bernstein function $f$ \cite{HanDuality,HanSerPRSA}. 

For $\tilde{L}(p) = a\, p^\alpha$, with $a, \alpha > 0$, the function $f$ 
assumes the form
$\theta(t) + 2 c_0 \,a\, t_+^{1-\alpha}/\Gamma(1-\alpha) + c_0^2\,a^2\, 
t_+^{2 ( 1-\alpha)}/\Gamma(3-2\alpha)$.
It is obvious that $f$ is a Bernstein function (and $G$ is completely monotone)
if and only if $1/2 \leq \alpha < 1$. We also note that 
$L(t) = a\, t_+^{-\alpha-1}/\Gamma(-\alpha)$ is a distribution.
Convolution with 
$L(t)$  is a Riemann-Liouville fractional differential operator of order 
$\alpha$. 
The order of the fractional differential equation \eqref{eq:motion} 
is $2$ if $\alpha \leq 1$ and $2 \alpha $ if $\alpha > 1$. (In the literature
the highest-order derivative $L\ast L\ast u$ is often incorrectly neglected).
For $\alpha > 1$ the order of the time-differential operator exceeds
the order of the spatial differential operator, hence the equation is 
formally parabolic. 

Concerning superlinear power law attenuation, note
that the logarithmic attenuation rate
\begin{equation}
\mathcal{A}(\omega) := \re B(-\ii \omega) = a \, \cos(\alpha\, \upi/2) \,
\vert \omega \vert^\alpha
\end{equation} 
is assumed non-negative, hence $1 < \alpha < 3$ must entail that $a < 0$. 
As the frequency tends to infinity the frequency-dependent phase  
speed $c(\omega)$, given by the equation  
$1/c(\omega) := \re [\ii B(- \ii \omega)] = 1/c_0 + 
a \, \vert \omega \vert^{\alpha-1} \, \sin(\alpha\,\upi/2)$
increases to its maximum value $c_0$ if $0 \leq \alpha < 1$ 
("abnormal dispersion"). 
For $1 \leq \alpha < 2$ it increases from $c_0$ at zero frequency to infinity
at a finite frequency 
$\omega_1 = 
\left[ c_0 \, \vert a \vert \, \sin(\alpha\, \upi/2) \right]^{1/(\alpha-1)}$ 
and changes sign. This behavior is clearly unphysical. 
For $2 < \alpha < 3$ phase speed decreases from $c_0$ at $\omega = 0$
to 0 at infinite frequency ("normal dispersion").

\section{Models with variable attenuation exponent.}

The limits on the dissipation-attenuation exponent $\alpha$ actually apply to 
the asymptotic value of $b(p)$ at infinity. Is it possible that 
the experimentally measured power law behavior applies to the middle frequency 
range?

Define the variable dissipation-attenuation exponent as the function
\begin{equation}
\alpha(p) = \ln(b(p))/\ln(p), \qquad p > 1
\end{equation}
so that $b(p) = p^{\alpha(p)}$. This definition has a major flaw:
a singularity at $p = 1$. The exponent decreases to $-\infty$ at 
$p \rightarrow 1-$ and restarts from $\infty$ to decrease towards its 
asymptotic values.

The simplest examples of dispersion-attenuation functions
with variable exponent are \begin{equation} 
b_1(p) = c\, p^\alpha/(a + p^\alpha),\qquad 0 \leq \alpha < 1
\end{equation}
and 
\begin{equation} 
b_2(p) = c\, (1 + \tau \,p)^{\alpha-\beta}\,(\tau\, p)^\beta, \qquad 0 < \beta \leq \alpha < 1 
\end{equation}
$b_1$ is a complete Bernstein function because $p/b_1(p) \equiv p/c + a\, p^{1-\alpha}/c$
is obviously a complete Bernstein function. Moreover $b_1(0) = 0$ and 
$\lim_{p\rightarrow\infty} b_1(p)/p = 0$. Hence $b_1$ is an admissible dispersion-attenuation
function. 

In order to prove that $b_2$ is a complete Bernstein function we need 
Theorem~\ref{thm:J}.
We now note that $\arg b_2(p) = \arg(1 + p)^{\alpha-\beta} + \arg p^\beta \leq
(\alpha-\beta)\, \arg p + \beta \, \arg p = \alpha \, \arg p$. For $p$ in the upper half-plane
$0 \leq \arg p \leq \upi$, hence $0 \leq \arg b_2(p) \leq \upi$ and  
$b_2$ is a complete Bernstein function. Since $b_2(0) = 0$ and 
$\lim_{p\rightarrow\infty} b_1(p)/p = 0$,  $b_2$ is admissible as a dispersion-relaxation
function. In the first case, $b = b_1$, the attenuation exponent $\alpha(p) \leq 0$
and $\alpha(p) \rightarrow 0$ from below as $p \rightarrow \infty$. 
In the second case $b = b_2$ and $\alpha(p)$ increases from $\alpha(0) = \beta$ 
to $\alpha(\infty) = \alpha$. It is thus likely that the exponent assumes
a larger value in the high frequency range. Thus a value of the exponent 
above 1 in the middle frequency range is not likely.

\section{Conclusions.}

The class of admissible dispersion and attenuation functions can be 
characterized by a class of Radon measures. The theory applies only
to sublinear attenuation growth in the high frequency range. 

Superlinear growth of the attenuation function in the high frequency range is
incompatible with the assumption of positive relaxation spectrum underlying
theoretical and experimental viscoelasticity. Superlinear growth of
attenuation also implies that the phase speed is unbounded for
high frequencies and the main signal is preceded by a precursor of 
infinite extent. 

Explanation of superlinear frequency dependence of attenuation observed 
in many real materials remains a challenging task.

\end{document}